%% file: main.tex
\documentclass[letterpaper]{article} 
\usepackage[preprint]{aaai2027}
\usepackage[hyphens]{url}  
\usepackage{graphicx} 
\usepackage{natbib}  
\usepackage{caption} 
\usepackage{amsmath, amssymb, amsthm, booktabs, xcolor}
\usepackage{algorithm}
\usepackage[noend]{algpseudocode}
\newtheorem{definition}{Definition}
\newtheorem{proposition}{Proposition}
\graphicspath{{bfa_figures/}}
\newcommand{\pizero}{$\pi_0$}
\title{Bit-Flip Attacks on Vision-Language-Action Models:\\ Action-Decoding Architecture Shapes the Vulnerability}
\input{author_block}

\begin{document}
\maketitle

\begin{abstract}
Quantized Vision-Language-Action (VLA) models expose a weight-fault surface: Rowhammer-style faults can corrupt deployed INT8 bits. We present the first bit-flip attack on a VLA: a few gradient-selected flips reduce closed-loop success to $0\%$, while hundreds of random flips are harmless. Across four model variants spanning three action-head families, damaging bits concentrate in a few action-generating layers, but the empirical budget depends sharply on the head: direct regression and token policies fall in $1$--$5$ flips, whereas the evaluated flow-matching policies require ${\sim}100$--$300$. Our fixed-direction manifold-escape loss cuts \pizero{}'s budget from ${\sim}1000$ to ${\sim}100$ flips, and a matched five-direction sweep shows that the attack is not specific to an all-positive direction. On a direct head, protecting $3.1\%$ of weights preserves $60\%$ success at $K{=}100$, and protecting $5.3\%$ moves the open-loop break threshold from 3 to 100 flips. Finally, task-calibrated emulated $K{=}100$ flips yield $0/20$ real-robot successes, versus $14/20$ clean and $16/20$ global-random. Weight integrity is therefore a security boundary for embodied foundation models. Code is included as ancillary material.
\end{abstract}

\section{Introduction}
Vision-Language-Action (VLA) models such as OpenVLA \citep{openvla}, \pizero{} \citep{pizerocite}, and RT-2 \citep{rt2} map images and instructions directly to robot actions. Existing attacks target inputs or training-time backdoors \citep{gao2024dual}, leaving deployed weights intact. Yet low-precision edge deployment \citep{openvla} makes integer weights a target for Rowhammer-style faults. A few flips have crushed quantized classifiers \citep{bfa,deephammer} and, more recently, LLMs \citep{flipllm,dnl}, but only at categorical outputs. A classifier fault is read out once; a VLA fault passes through an action decoder and then repeatedly interacts with a changing environment. Decoder dynamics may attenuate the immediate perturbation, while modest residual errors can compound through feedback. Classifier BFA budgets and open-loop deviations therefore do not predict closed-loop robot failure. We ask: how few weight flips cause such failure, and how can they be contained?

\begin{figure}[t]\centering\includegraphics[width=\linewidth]{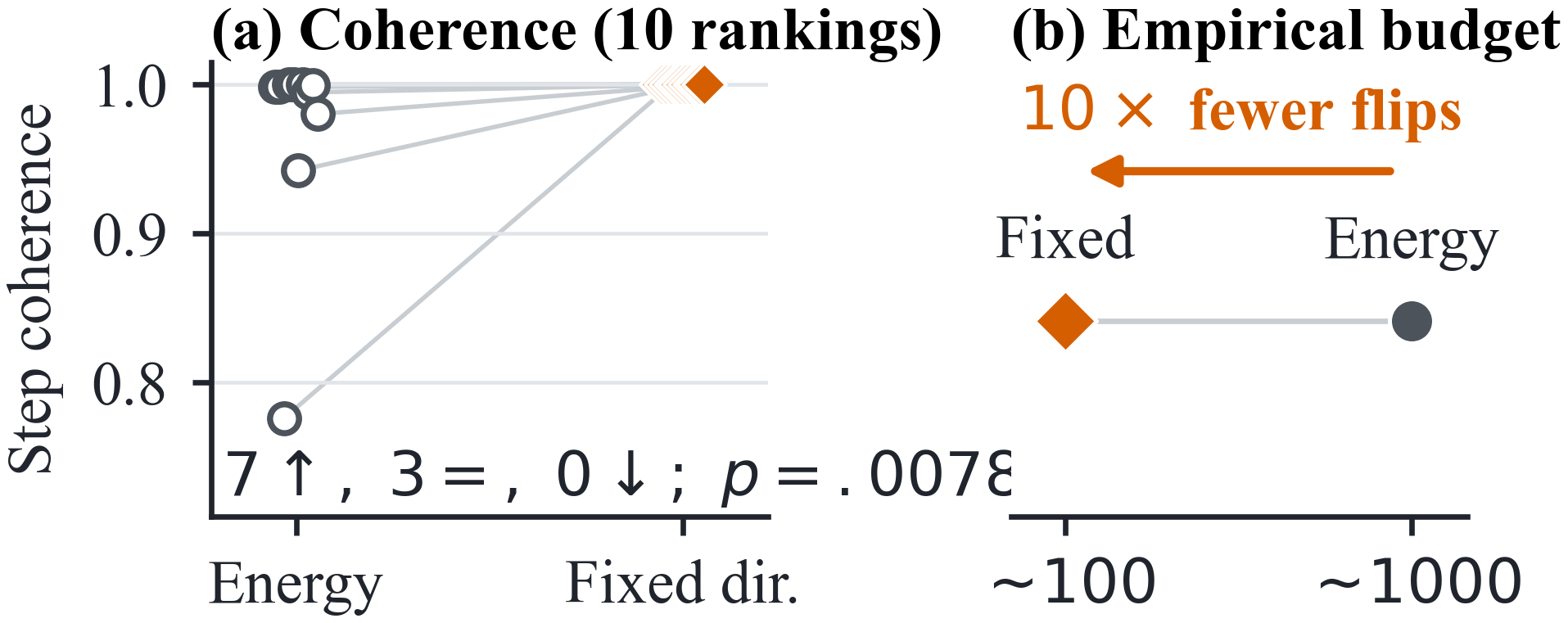}
\caption{\textbf{Fixed-direction insight.} Across ten rerankings it is never less coherent than energy (7 higher, 3 ties) and lowers \pizero{}'s empirical collapse budget \(10\times\).}\label{fig:insight}\end{figure}

On LIBERO-Spatial, 3 selected flips collapse simulated closed-loop success from 88\% to 0\% ($n{=}50$); $1$--$5$ flips suffice across four suites, while 300 random flips are harmless (Table~\ref{tab:sr}). The evaluated flow heads instead require ${\sim}100$--$300$ flips. On a real 6-DoF arm, task-calibrated emulated $K{=}100$ perturbations yield $0/20$ successes, versus $14/20$ clean and $16/20$ global-random (\S\ref{ss:attack}).

Figure~\ref{fig:insight} motivates our flow-head objective: across ten independent \(K{=}100\) rerankings, fixed direction is never less coherent than energy (seven higher, three ties; one-sided Wilcoxon \(p{=}.0078\)) and its same-policy empirical collapse budget is \(10\times\) lower. We therefore optimize the same nonzero, one-sided executed-action projection on every calibration trace, seeking perturbations with a common action-space orientation. Backpropagating this objective through the complete decoder scores each INT8 bit by its quantization-aware directional gain.

The complementary \emph{where} insight is spatial: damaging bits concentrate in a few action-generating layers (Fig.~\ref{fig:soft}). Together, these observations guide both attack and defense. Gradient ranking discovers high-leverage layers, and a matched sweep confirms that fixed-direction manifold escape is not peculiar to the all-positive default (\S\ref{ss:flowanalysis}). Conversely, protecting L1 ($3.1\%$) preserves closed-loop success through $K{=}100$; adding the action head ($5.3\%$ total) moves the open-loop threshold from 3 to 100.

\paragraph{Contributions.}
\textbf{(1)~The first weight-fault attack on a VLA.} To our knowledge, this is the first bit-flip attack on a VLA and a continuous, closed-loop action policy. Just $1$--$5$ selected INT8 flips collapse simulated task success to $0\%$, revealing closed-loop consequences absent from prior attacks on categorical outputs.
\textbf{(2)~An empirical vulnerability spectrum and an architecture-aware attack.} Across direct-regression, discrete-token, and flow-matching decoders, we localize damaging bits to high-leverage action layers. Fixed-direction manifold escape breaks the evaluated flow heads with ${\sim}10\times$ fewer flips than isotropic energy, and a matched five-direction sweep rejects an all-positive-direction artifact.
\textbf{(3)~Broad evaluation and localized protection.} Across five simulation checkpoints, four LIBERO suites, and SimplerEnv, attack budgets range from $1$--$5$ flips for direct heads to ${\sim}100$--$300$ for flow heads. A separate real-trained $\pi_{0.5}$ with task-calibrated emulated $K{=}100$ perturbations yields $0/20$ real-robot successes. On a direct head, $3.1\%$ protection preserves $60\%$ success at $K{=}100$, and $5.3\%$ protection raises the open-loop threshold $33\times$.

\section{Related Work}

\paragraph{Security of VLAs and generative robot policies.} AttackVLA and ANNIE study input-side VLA attacks \citep{attackvla,annie}; BadVLA, TrojanRobot, and DropVLA study training- or supply-chain backdoors, including action-level backdoors \citep{badvla,trojanrobot,dropvla}. Related work also exposes oracle-level integrity attacks on the imagined trajectories \citep{chen2026attacking}. For diffusion and flow policies \citep{diffusionpolicy,flowmatching}, DP-Attacker perturbs inputs and identifies the visual encoder as vulnerable \citep{dpattacker}, while TrojFlow implants triggers \citep{trojflow}. Broader visual-security work studies stealthy backdoors, retraining-free backdoor removal, black-box reconstruction defenses, and input-triggered manipulation of 3D Gaussian-splatting models \citep{gao2024dual,gao2024energy,yu2025blackbox,,lu2026privacy,wu2025remisvfu,gao2026ghostsplat}. We instead corrupt served weights: damage concentrates in the action expert, and the evaluated flow heads require larger directed budgets (\S\ref{ss:underbelly}--\ref{ss:flowanalysis}).

\paragraph{Bit flips and quantized deployment.} BFA/PBS gradient-ranks INT8 bits, collapsing classifiers with ${\sim}11$--$17$ flips while hundreds of random flips do little \citep{bfa}; targeted variants install chosen behavior \citep{tbfa}. Recent work scales search to LLMs \citep{flipllm}, finds sparse sign-bit lesions that motivate selective protection \citep{dnl}, and reports block-localized BFA sensitivity in vision transformers \citep{vitbfa}. DeepHammer realizes DDR3/4 faults \citep{deephammer}, while GPUHammer and GDDRHammer extend Rowhammer to GPU memory \citep{gpuhammer,gddrhammer}. Defenses include binarization, clipping, and selective checksums/ECC \citep{rabnn,bfadefense,radar,bitshield}. Complementary model protection addresses ownership verification through expert-routing path watermarks in MoE LLMs \citep{gao2026pathmark}. VLA quantization studies efficiency \citep{bitvla,quantvla,actquant}, while quantization-conditioned backdoors poison weights before release \citep{quantibackdoor}; neither addresses served-weight faults in continuous closed-loop action.

\section{Preliminaries}\label{sec:prelim}
\paragraph{Quantized VLA policies.}
A VLA policy $\pi_\theta:\mathcal{O}\!\to\!\mathcal{A}$ maps $o=(\text{image},\text{instruction})$ to an $H$-step, $d$-dimensional action chunk $a_\theta(o)$. Its action-generating transformer has weights $\theta=\{W^{(\ell)}\}_{\ell}$. We evaluate three decoders: a \emph{direct-regression} head, $a_\theta=W_oh_\theta(o)+b_o$; a \emph{discrete-token} head with $C=256$ bins, $a_{h,r}=\operatorname{decode}(\arg\max_c z_{h,r,c}(o))$; and a \emph{flow-matching} head that transforms $x_1\sim\mathcal{N}(0,I)$ into $a=x_0$:
\begin{equation}
\begin{aligned}
a_\theta(o,x_1)&=x_1+\int_1^0 v_\theta(x_t,t,c_o)\,dt\\
&\approx x_1+\sum_{k=1}^{N}\Delta t\,v_\theta(x_{t_k},t_k,c_o),
\qquad \Delta t=-1/N .
\end{aligned}
\label{eq:pi0}
\end{equation}
The final flow action therefore depends on a sequence of solver updates rather than a single readout.

\paragraph{INT8 logical faults.}
Every eligible linear weight is quantized per output channel; stored integers are faultable, while scales and activations remain intact.
\begin{definition}[Per-channel symmetric INT8 quantization]\label{def:quant}
For $W\in\mathbb{R}^{m\times n}$, let $s_i=\max\{\max_j|W_{ij}|/127,10^{-8}\}$ and $q_{ij}=\operatorname{clip}_{[-128,127]}(\operatorname{round}(W_{ij}/s_i))$. The stored $q_{ij}$ uses 8-bit two's-complement representation, and inference dequantizes $\widehat W_{ij}=q_{ij}s_i$. An all-zero row maps to $q_{ij}=0$.
\end{definition}
For the OFT policy, clean quantization changes the action by only $\|a_{\widehat\theta}(o)-a_\theta(o)\|=0.003$, well below its action scale.
\begin{definition}[Bit flip and induced weight perturbation]\label{def:flip}
Toggling $b\in\{0,\ldots,7\}$ maps $q_{ij}\mapsto\operatorname{flip}_b(q_{ij})$ and changes the dequantized weight by $\Delta_{ijb}=(\operatorname{flip}_b(q_{ij})-q_{ij})s_i$.
\end{definition}
Non-sign and sign flips change the integer by at most $64$ and $128$, respectively, so $|\Delta_{ijb}|\le128s_i$ and no INT8 flip directly produces $\infty$/NaN. Decoder-dependent propagation is therefore meaningful for these bounded faults.

\section{Threat Model}
\paragraph{Deployment and attacker.}
We consider an edge-served INT8 VLA whose model memory may be shared with an untrusted process. Following BFA \citep{bfa}, the attacker knows the architecture, stored weights, scales, and decoder, and has a small calibration set $\mathcal{D}_{\mathrm{cal}}$ disjoint from evaluation. It induces at most $K$ persistent linear-weight flips, one per selected scalar, but controls neither training, inputs, activations, scales, nor the environment. The baseline deployment lacks effective end-to-end memory integrity or ECC.

\paragraph{Goal and defender.}
The primary goal is to collapse closed-loop success with few flips; targeted diagnostics instead drive a chosen action. Selective protection removes defended weights from the candidate set and rebuilds the attack over the remainder, making defense evaluation adaptive rather than a replay.

\begin{figure*}[t]\centering\includegraphics[width=0.99\textwidth]{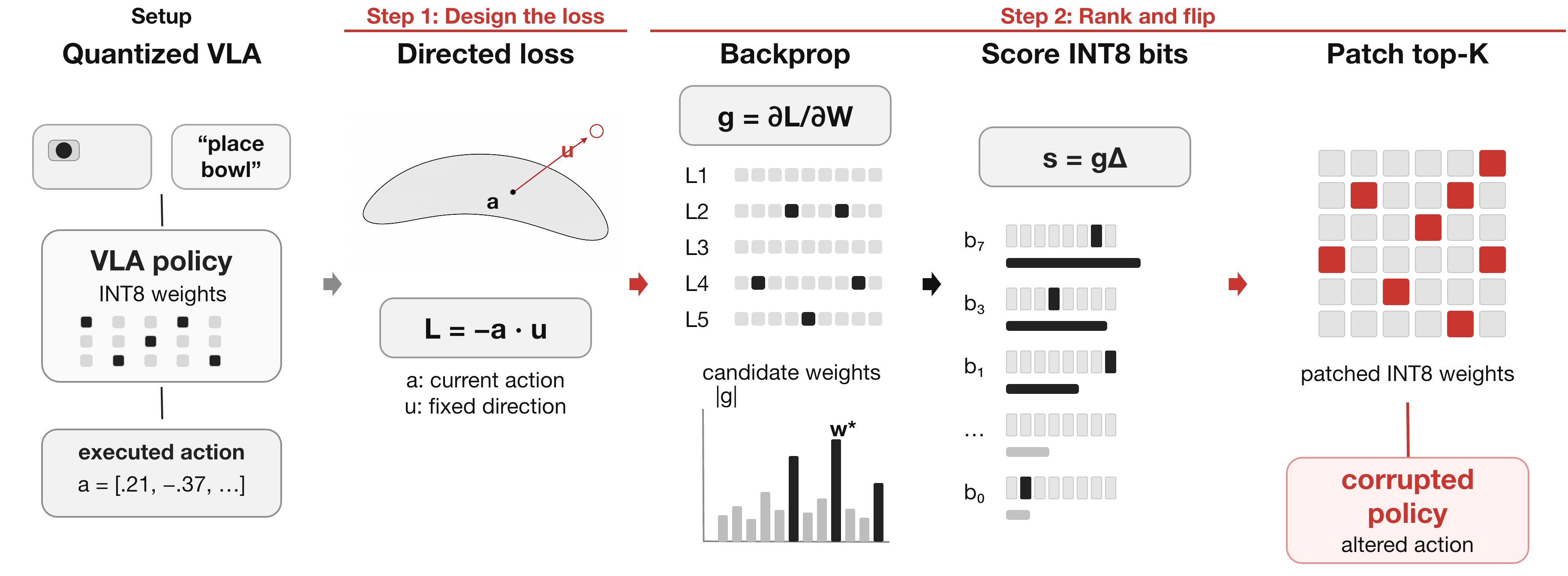}
\caption{\textbf{Fixed-direction manifold-escape attack.} The displayed $L=-a\cdot u$ is gradient-equivalent to the centered loss in Eq.~\eqref{eq:directional}. A bit score $s=g\Delta$ predicts its first-order loss change; ranking uses the corresponding gain $G=-s$ and discovers high-leverage layers without an architectural prior.}\label{fig:method}\end{figure*}

\paragraph{Physical-delivery boundary.}
DeepHammer maps profiled cells and induces targeted chains \citep{deephammer}; GPUHammer reports 8 flips across four GDDR6 banks, while GDDRHammer demonstrates cross-component attacks \citep{gpuhammer,gddrhammer}. Prior demonstrations reach fault counts comparable to our $1$--$5$ direct-head regime, but coordinate reachability remains device-, placement-, and fault-profile-dependent. We therefore evaluate logical INT8 corruption rather than end-to-end physical delivery; the ${\sim}100$--$300$ flow-head regime is not a deliverability claim.

\section{Method}
Our evidence points to two factors---how faults propagate through the decoder and where high-leverage weights lie---which become two design choices: an executed-action objective that remains effective through iterative decoding, and an unrestricted quantization-aware ranking that reveals sensitive locations. Step 1 specifies how the executed action should move, while Step 2 determines where that displacement can be induced most efficiently. The objective is decoder-aware; the location search remains architecture-agnostic. For $(\text{weight},\text{bit})$ pairs $\mathcal{F}$, define
\begin{equation}
\begin{aligned}
\delta\theta(\mathcal{F})
&=\sum_{((\ell,i,j),b)\in\mathcal{F}}
\Delta^{(\ell)}_{ijb}e^{(\ell)}_{ij},\\
\mathcal{F}^{\star}
&\in\arg\min_{|\mathcal{F}|\le K}
L\big(\widehat\theta+\delta\theta(\mathcal{F});\widehat\theta\big),
\end{aligned}
\label{eq:attack}
\end{equation}
where $e^{(\ell)}_{ij}$ is the unit weight-coordinate and the clean-reference argument is omitted when unnecessary. For $N_w$ eligible scalar weights, exact optimization over ${\sim}8N_w$ bit candidates is intractable, so we use two steps (Fig.~\ref{fig:method}).

\noindent\emph{Step 1: design the loss} (\S\ref{ss:obj}). We differentiate the final executed action for every decoder and construct a fixed-direction manifold-escape loss that remains effective through flow sampling.

\noindent\emph{Step 2: rank and flip} (\S\ref{ss:search}). Following gradient-ranked bit search \citep{bfa}, we combine the action gradient with each exact dequantized bit change. All action-generating linears remain eligible, so localization is an outcome rather than a prior. Algorithm~\ref{alg:pbs} gives the procedure.

\subsection{Step 1: a fixed-direction manifold-escape loss}\label{ss:obj}
\paragraph{Differentiable executed actions.}
We construct an attack-time differentiable action $\widetilde a_\theta$ for each decoder. Direct regression uses $\widetilde a_\theta=a_\theta$. For discrete tokens, only the ranking gradient replaces $\arg\max$ with the expected decoded value $\widetilde a_{h,r}=\sum_c\operatorname{softmax}(z_{h,r})_c\nu_c$, where $\nu_c$ denotes bin $c$; evaluation retains the hard decoder. A flow policy uses $\widetilde a_\theta=a_\theta=x_0$, backpropagated through all solver steps in Eq.~\eqref{eq:pi0}. Because the attack targets executed action rather than field fit, we use this sampling gradient and average over $M=2$--$3$ fixed noise draws; deterministic heads use $M=1$.

\paragraph{Centered directional objective.}
Several natural ports fail on the evaluated flow head. Direct squared deviation, $-\|\widetilde a_\theta-\widetilde a_{\widehat\theta}\|_2^2$, has zero gradient at $\theta=\widehat\theta$ and cannot initialize first-order search. Driving a fixed valid action saturates empirically, while isotropic energy requires ${\sim}10\times$ more flips (\S\ref{ss:flowanalysis}). We instead detach the clean action and optimize a one-sided displacement:
\begin{definition}[Fixed-direction manifold-escape objective]\label{def:obj}
For a nonzero direction $u\in\mathbb{R}^{H\times d}$, define
\begin{equation}
\begin{aligned}
d_\theta^{(m)}(o)
&=\widetilde a_\theta(o,x_1^{(m)})
-\operatorname{sg}\!\left[\widetilde a_{\widehat\theta}(o,x_1^{(m)})\right],\\
L_{\mathrm{dir}}(\theta;\widehat\theta)
&=-\frac{1}{|\mathcal{D}_{\mathrm{cal}}|M}
\sum_{o\in\mathcal{D}_{\mathrm{cal}}}\sum_{m=1}^{M}
\langle d_\theta^{(m)}(o),u\rangle ,
\end{aligned}
\label{eq:directional}
\end{equation}
where $\operatorname{sg}$ denotes stop-gradient and the noise input is omitted for deterministic heads.
\end{definition}
Although $L_{\mathrm{dir}}(\widehat\theta;\widehat\theta)=0$, detaching the reference gives the same nonzero gradient as $-\langle\widetilde a_\theta,u\rangle$ while making clean-action displacement explicit. We use $u=\mathbf{1}$ and test $-\mathbf{1}$ plus three Gaussian directions in \S\ref{ss:flowanalysis}; positive rescaling does not change ranking. Sharing $u$ gives traces a common action-space orientation without fitting a geometric manifold or forcing per-step gradient signs. Targeted diagnostics instead use $L_{\mathrm{tgt}}=\mathbb{E}_o\|\widetilde a_\theta(o)-a^\star\|_2^2$.

\subsection{Step 2: quantization-aware bit ranking}\label{ss:search}
Given $L_{\mathrm{dir}}$, which bits best solve Eq.~\eqref{eq:attack}? Let
\begin{equation}
\begin{aligned}
g^{(\ell)}_{ij}
&=\left.\frac{\partial L_{\mathrm{dir}}}
{\partial W^{(\ell)}_{ij}}\right|_{\theta=\widehat\theta},\\
s^{(\ell)}_{ijb}
&=g^{(\ell)}_{ij}\Delta^{(\ell)}_{ijb}
\approx \Delta L^{(\ell)}_{ijb},\qquad
G^{(\ell)}_{ijb}=-s^{(\ell)}_{ijb}.
\end{aligned}
\label{eq:rank}
\end{equation}
The score $s$ predicts the loss change; $G$ is positive when a flip improves the minimizing attack. Since $L_{\mathrm{dir}}$ is a fixed action projection,
\begin{equation}
G^{(\ell)}_{ijb}
=\Delta^{(\ell)}_{ijb}\,
\mathbb{E}_{o,m}\!\left[
\left\langle
\frac{\mathrm{d}\widetilde a_\theta}
{\mathrm{d}W^{(\ell)}_{ij}},u
\right\rangle\right]_{\theta=\widehat\theta}.
\label{eq:gsens}
\end{equation}
Thus $G$ combines directional action sensitivity with the quantization-dependent flip magnitude; it is a projection, not the full sensitivity norm. For each scalar weight we keep $b^\star_{ij}=\arg\max_b G_{ijb}$, then select the $K$ largest positive per-weight gains. Directional leverage is highly nonuniform, so this unrestricted ranking concentrates in a few action-generating layers without being told where to search (\S\ref{ss:underbelly}).

The search requires one backward pass and scores eight bits per eligible weight, giving $O(8N_w)=O(N_w)$ scoring cost. Its one-shot gradient ignores interactions among selected flips; progressive re-ranking can capture some interactions at higher cost, and \S\ref{ss:attack} verifies that the low-budget discrete result survives canonical PBS.

\begin{algorithm}[t]
\caption{Quantization-aware fixed-direction bit search}\label{alg:pbs}
\begin{algorithmic}[1]
\Require quantized weights $\{q^{(\ell)},s^{(\ell)}\}$, candidates $\mathcal{C}$, calibration set $\mathcal{D}_{\mathrm{cal}}$, direction $u$, noise draws $M$, budget $K$
\State $g\gets\left.\nabla_W L_{\mathrm{dir}}(\theta;\widehat\theta)\right|_{\theta=\widehat\theta}$
\For{each weight $(\ell,i,j)\in\mathcal{C}$}
  \State compute $\Delta^{(\ell)}_{ijb}$ and $G^{(\ell)}_{ijb}=-g^{(\ell)}_{ij}\Delta^{(\ell)}_{ijb}$ for $b=0,\ldots,7$
  \State $b^{\star(\ell)}_{ij}\gets\arg\max_bG^{(\ell)}_{ijb}$,\quad $G^{(\ell)}_{ij}\gets G^{(\ell)}_{ijb^{\star(\ell)}_{ij}}$
\EndFor
\State $\mathcal{F}\gets$ top-$K$ positive-gain weights paired with $b^\star$
\State \Return $\mathcal{F}$ and $\widehat\theta+\delta\theta(\mathcal{F})$
\end{algorithmic}
\end{algorithm}
The same algorithm collapses the discrete policy at $K{=}3$, drives OFT's action a full scale off, and, with the primary $u=\mathbf{1}$, collapses \pizero{} at $K{\approx}100$: one recipe whose empirical budget varies sharply with the action decoder.

\section{Experiments}
\noindent\textbf{Overview.} Quantization makes the policy look safe: naive bit-flips leave the action essentially unchanged. Yet a gradient-ranked search that maximizes a fixed executed-action projection collapses closed-loop task success to zero, and the required budget spans two orders of magnitude across architectures. We establish the attack across this spectrum (\S\ref{ss:attack}), show that protecting ${\sim}3$--$5\%$ of weights substantially raises the tested budget of a direct head (\S\ref{ss:defense}), and analyze localization and flow-head objectives (\S\ref{ss:underbelly}--\ref{ss:flowanalysis}).

\subsection{Experimental setup}\label{sec:setup}
\textbf{Models.} We study three action-head families across four simulation model variants and five simulation checkpoints ($3$--$7$B, the $7$B OpenVLA family and the ${\sim}3.3$B \pizero{}/$\pi_{0.5}$): (i) OpenVLA-OFT \citep{openvlaoft} (continuous $L_1$-regression head, 8-step chunks), (ii) discrete OpenVLA (256-bin action tokens emitted by the LLM head), (iii) \pizero{} (flow-matching with a separate action expert, $N{=}10$ denoising steps), and (iv) $\pi_{0.5}$ \citep{pi05cite} (a newer flow-matching VLA with a distinct architecture and weights, used to test whether the flow-matching findings generalize beyond a single model). (i)--(iv) are LIBERO-finetuned. For the cross-benchmark transfer we additionally use (v) the OXE-pretrained \texttt{openvla-7b} base policy (discrete token head, Google-Robot/Bridge embodiment). The real-robot study separately uses a real-trained $\pi_{0.5}$ checkpoint.

\textbf{Benchmarks.} Closed-loop evaluation uses LIBERO \citep{libero} (four suites: Spatial, Object, Goal, Long) and SimplerEnv \citep{simplerenv} (ManiSkill2/SAPIEN, Google-Robot pick-coke-can and move-near).

\textbf{Quantization and attack surface.} Weights are quantized per Def.~\ref{def:quant}. The attacker's candidate set is every linear layer of the action-generating transformer (all decoder layers plus the action head/expert, $228$--$437$ matrices), a realistic Rowhammer surface rather than the small head alone. The ranking gradient is accumulated over ${\le}6$ calibration frames, disjoint from evaluation. The attack is robust to both knobs: a single calibration frame already collapses the discrete policy at $K{=}3$ (identical to $6$ or $20$ frames), and either extreme direction works (push-high $0\%$, push-low $5\%$ at $K{=}3$, $n{=}20$), so the attack does not depend on tuned hyperparameters.

\textbf{Metrics.} (a) \emph{Open-loop action deviation}: mean per-step $L_2$ between corrupted and clean executed actions on held-out frames (baseline action scale $0.804$ OFT $/$ $0.169$ discrete $/$ $0.638$ \pizero{}). (b) \emph{Closed-loop success rate} (SR): the fraction of successful rollouts. Headline endpoints use $n{=}30$--$50$ with Clopper-Pearson (exact binomial) $95\%$ intervals on $0/n$ collapses, and the random-flip control reports a 5-seed mean$\pm$std. (c) \emph{Safety predicates}: per-step gripper inversion, out-of-envelope motion, and direction reversal for targeted diagnostics. Each closed-loop cell aggregates $\mathrm{NT}$ tasks $\times\,\mathrm{EP}$ episodes. All simulated experiments use one NVIDIA A800; the physical-robot runtime is described separately in \S\ref{ss:attack}. Clean SR varies across experiments (discrete $88$--$90\%$, \pizero{} $57$--$70\%$) because attacks use different small task--episode subsets; every attack is compared with the matched clean baseline from the same subset.

\textbf{Reproducibility.} The appendices document the statistical protocol and additional diagnostics; ancillary code provides the attack, evaluation, analysis, and real-robot patching implementations. The \pizero{} direction sweep fixes seed 20260722, three calibration frames, two ranking-noise draws, and identical 30 rollout slots across directions. The fixed-path study independently reranks ten seeds (20260730--39), each with 17 held-out frames and three shared sampler-noise draws. These jobs use Ubuntu 22.04, Python 3.12, PyTorch 2.5, Transformers 4.53, LeRobot 0.4.4, and LIBERO 0.1.0; checkpoints, benchmark assets, calibration observations, robot videos, and experiment outputs are not redistributed.

\subsection{The attack}\label{ss:attack}
\paragraph{Na\"ive flips (premise).}
In bf16 a single magnitude-targeted exponent-MSB flip already moves the OFT action by 0.306 (38\% of scale) and ${\ge}10$ random flips produce NaNs (a denial-of-service), but this is a floating-point artifact. In the realistic INT8 regime, bounded values eliminate NaNs and na\"ive attacks fail: 300 random and 100 magnitude-targeted flips both leave the action within ${\sim}0.004$ of clean (Fig.~\ref{fig:premise}), mirroring the BFA literature. Does INT8 also resist a gradient-ranked search?
\begin{figure}[t]\centering\includegraphics[width=\linewidth]{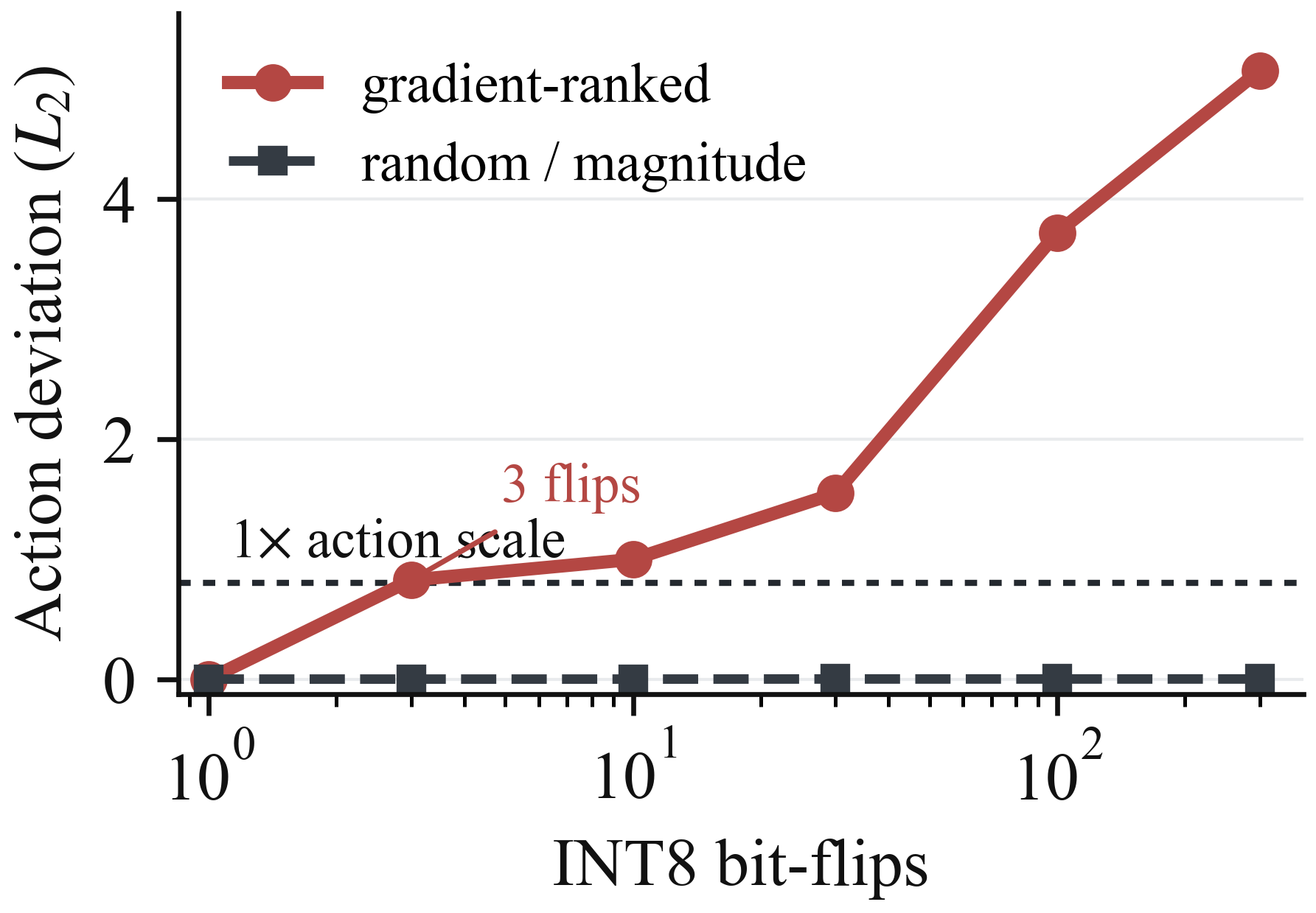}
\caption{Three gradient-ranked INT8 flips move OpenVLA-OFT a full action scale; 300 na\"ive flips barely move it.}\label{fig:premise}\end{figure}

\paragraph{Gradient-ranked bit search.}
Full-model gradient-PBS on OFT, over all 228 LLM-decoder + action-head linears: $K{=}3\to$ deviation 0.828 (${\approx}$full scale), $K{=}10\to1.00$, $K{=}100\to1.13$, three orders of magnitude above na\"ive at matched budgets. The selected flips are dominated by LLM Layer 1 (earliest decoder layer) and the action head, the highest-leverage points on this direct head's action pathway.

\paragraph{Closed-loop task failure.}
Closed-loop LIBERO rollouts of discrete OpenVLA yield the suite-level results below:
\begin{table}[t]\centering\small
\setlength{\tabcolsep}{2.8pt}
\begin{tabular}{lcccc}\toprule
suite & clean & min $K$ & 95\% upper & rand.-300\\\midrule
Spatial ($n{=}50$) & 88.0\% & 3 & ${\le}7.1\%$ & 80.0\%\\
Object ($n{=}30$) & 80.0\% & 5 & ${\le}11.6\%$ & n/a \\
Goal ($n{=}30$) & 83.3\% & 2 & ${\le}11.6\%$ & n/a \\
Long ($n{=}30$) & 50.0\% & 1 & ${\le}11.6\%$ & n/a \\\bottomrule
\end{tabular}
\caption{Closed-loop attack results on four LIBERO suites.}\label{tab:sr}\end{table}
The discrete robot fails every task after 3 gradient-selected flips on LIBERO-Spatial ($0/50$ successes). The random-flip control is decisive: 300 random flips leave success near clean (80\% vs.\ 88\%, over 5 seeds $90.0\%\pm5.4\%$, never collapsing), confirming the loss-aligned search, not the flip count, is the cause. Three flips induce only a modest per-step deviation (0.105) yet drive closed-loop success to zero: errors compound over the ${\sim}200$-step horizon. Consistent with horizon compounding, the longest-horizon suite (LIBERO-Long) collapses at a single flip (Table~\ref{tab:sr}). The continuous-regression OFT head is even more fragile: a single gradient flip in its small, high-leverage action head collapses closed-loop success. Not all architectures are equally fragile, though: the flow-matching \pizero{} requires a different objective and a larger budget (\S\ref{ss:flowanalysis}).

\paragraph{Progressive-ranking check.}
Although the main attack ranks once, we also run canonical PBS with a fresh gradient and ranking after every discrete-policy flip. Open-loop deviation is $0.017$, $0.090$, and $0.153$ at $K{=}1,2,3$, respectively, and every selected bit remains in LLM L1. Thus the rapid low-budget damage and early-layer localization are not artifacts of a frozen ranking.

\paragraph{All-linear weight quantization and lower precision.}
Table~\ref{tab:quant} summarizes the deployment-quantization checks. Quantizing every linear weight to INT8 preserves clean behavior and the direct-versus-flow budget gap. Lower precision does not remove the threat: per-channel INT4 reduces the discrete collapse budget to two flips, while usable group-wise INT4 preserves \pizero{}'s ${\sim}100$-flip regime. For \pizero{}, per-channel INT4 reduces clean SR to 10.0\% as denoising accumulates quantization error, so we exclude it from attack-budget comparisons.
\begin{table}[t]\centering\small
\setlength{\tabcolsep}{4.5pt}
\begin{tabular}{llcc}\toprule
policy & stored-weight quantizer & clean SR & break \(K\)\\\midrule
discrete & INT8, all 437 linears & 88.9\% & 3\\
\pizero{} & INT8, all 422 linears & 56.7\% & ${\sim}100$\\
discrete & INT4, per-channel & 76.7\% & 2\\
\pizero{} & INT4, group-wise & preserved & ${\sim}100$\\
\pizero{} & INT4, per-channel & 10.0\% & n/a\\\bottomrule
\end{tabular}
\caption{Attack budgets under INT8 and INT4 quantization.}\label{tab:quant}
\end{table}
For discrete INT4, SR drops from 76.7\% clean to 46.7\% at \(K{=}1\) and 0\% at \(K{=}2\). These comparisons isolate stored-weight precision while ranking and arithmetic remain floating point.

\paragraph{Cross-benchmark transfer.} We also attack a different benchmark, embodiment, and checkpoint: the OXE-pretrained \texttt{openvla-7b} base policy on SimplerEnv (ManiSkill2/SAPIEN, Google Robot). On pick-coke-can, SR is $20.8\%\to0\%$ at $K{=}3$ (random-300: $29.2\%$). On the stronger move-near task it is $58.3\%\to12.5\%$ at $K{=}3$, $4.2\%$ at $K{=}5$, and $0\%$ at $K{=}10$ (random-300: $41.7\%$), all at $n{=}24$ (Table~\ref{tab:simpler}). The higher-clean move-near task provides the stronger evidence; together, the results demonstrate transfer across benchmark, embodiment, and OXE- versus LIBERO-finetuned checkpoints.
\begin{table}[t]\centering\small
\begin{tabular}{lccc}\toprule
SimplerEnv task & clean & flips $\to$ 0\% & random \\\midrule
pick-coke-can & 20.8\% & 3 & 29.2\% \\
move-near & 58.3\% & 10 & 41.7\% \\\bottomrule
\end{tabular}
\caption{SimplerEnv attack results ($n{=}24$ per task).}\label{tab:simpler}
\end{table}

\paragraph{Targeted action diagnostics.}
On held-out traces of the discrete policy, a target-action variant concentrates at low $K$ on the most-sensitive gripper coordinate: $K{=}3$ inverts the gripper command on $99\%$ of steps while leaving arm-motion deviation at ${\sim}0.003$. At $K{=}30$, $60\%$ of commands leave the clean envelope and $41\%$ reverse direction. We report these outcomes as per-step open-loop safety diagnostics; closed-loop targeted behavior is outside this evaluation.

\paragraph{Trigger-conditioned stealthiness (negative result).}
We additionally test a T-BFA-style goal: preserve clean behavior but alter it under a checkerboard trigger. The trigger separates clean and triggered attack gradients (cosine ${\approx}0.24$) while leaving patch-only success near clean. An iterative search with a clean-preservation penalty keeps the clean action-token proxy within ${\pm}3$ bins and drives the triggered proxy away over 30 flips, yet clean closed-loop SR still falls to $0\%$ at $K{=}10$--$30$. Thus we do not obtain a stealthy triggered backdoor.

\paragraph{Real-robot study.} On a 6-DoF place-block-in-bowl task, a collaborator patched a real-trained $\pi_{0.5}$ with task-calibrated dequantized values equivalent to \(K{=}100\) INT8 flips (99 readout coordinates, one expert). The patch yields $0/20$ successes versus $14/20$ clean and $16/20$ equal-count global-random (Fig.~\ref{fig:realrobot}; exact 95\% CIs: $[0,16.8]\%$, $[45.7,88.1]\%$, and $[56.3,94.3]\%$). Two-sided Fisher tests give $p<10^{-5}$ against either control; clean and random do not differ ($p=.716$). The directed list contains 100 distinct weight coordinates: 99 in \texttt{action\_out\_proj} and one in expert-L17, with 96 sign-bit and four bit-6 flips. The global-random list also uses 100 distinct coordinates but spans all eight bit positions. This study validates task-calibrated logical weight corruption on a real robot using an equal-count global-random control; end-to-end fault delivery and cross-task transfer remain open.
\begin{figure}[t]\centering\includegraphics[width=\columnwidth]{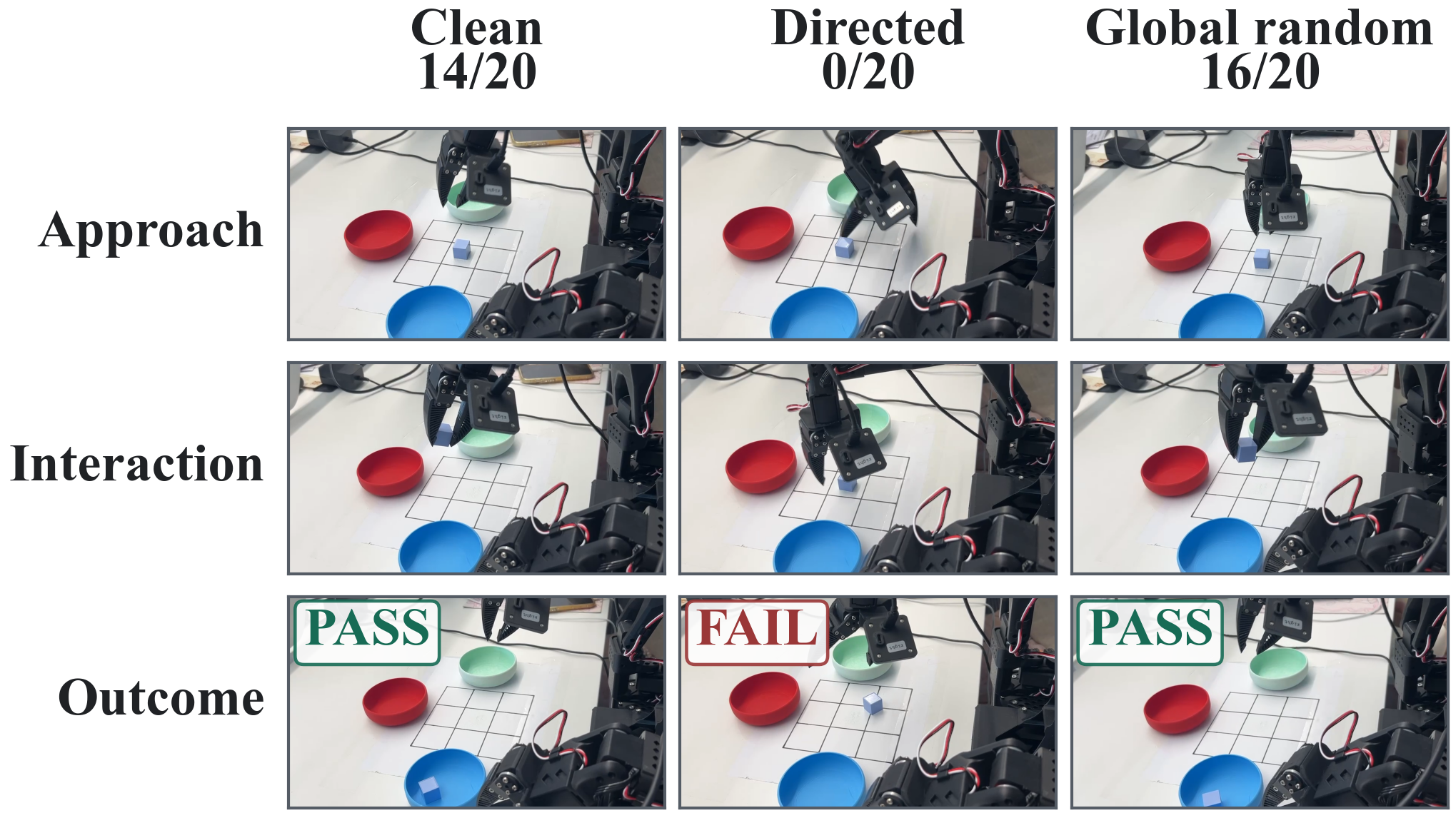}
\caption{Physical-robot rollouts on the calibrated blue-bowl task. Headers report 20-trial success; rows show stages under emulated INT8-equivalent patches.}\label{fig:realrobot}\end{figure}

\subsection{Defense: Localized Integrity Protection}\label{ss:defense}
We protect weights by removing them from the candidate set and rebuilding the attack, so every result is adaptive rather than a replay of fixed flips. Open loop, protecting the intuitive action head alone is useless: the attack pivots to L1 and still breaks at $K{=}3$ (Table~\ref{tab:defopen}). Adding L1 raises the threshold to $K{=}100$ ($33\times$), while expanding coverage to L0--5 provides no further gain.
\begin{table}[t]\centering\small
\begin{tabular}{lccc}\toprule
protection & weights & break $K$ & dev.@$K{=}3$\\\midrule
none & 0.0\% & 3 & 0.83\\
action head & 2.3\% & 3 & 0.75\\
head + L1 & 5.3\% & 100 & 0.14\\
L0--5 & 20.6\% & 100 & 0.12\\\bottomrule
\end{tabular}
\caption{Adaptive open-loop defense (break: dev.\ ${>}0.5$).}\label{tab:defopen}
\end{table}

Closed loop requires only L1 (3.1\%): success remains 80\%, 83\%, and 60\% at $K{=}3,10,100$ (Table~\ref{tab:def}). At matched \(K{=}3\), specific protection yields \(24/30\) successes versus \(0/30\) unprotected and \(3/30\) for an equal-size random slice (two-sided Fisher \(p=3.3{\times}10^{-11}\) and \(5.6{\times}10^{-8}\), respectively). After L1 is excluded, the rebuilt ranking redistributes to later layers (L9, L19, and L23), yet \(18/30\) trials still succeed at \(K{=}100\) (exact 95\% CI: \([40.6,77.3]\%\)). Survival through \(K{=}100\) therefore establishes an at-least \(33\times\) increase in the tested attack budget.
\begin{table}[t]\centering\small
\setlength{\tabcolsep}{3.0pt}
\begin{tabular}{lcccc}\toprule
protection & weights & SR@3 & SR@10 & SR@100\\\midrule
none & 0.0\% & 0\% & 0\% & --\\
early layer (L1) & 3.1\% & 80\% & 83\% & 60\%\\
random slice & 3.0\% & 10\% & n/a & -- \\\bottomrule
\end{tabular}
\caption{Adaptive closed-loop defense (clean 90\%, $n{=}30$).}\label{tab:def}
\end{table}
\paragraph{Scope and full-cost alternatives.} Localized protection does not transfer to \pizero{}: protecting 14\% of its expert fails as the attack redistributes, so flow heads need broader coverage but already require ${\sim}30\times$ more flips. Full checksums or ECC are the 100\%-coverage endpoint of our defense. For a 7B-parameter INT8 model, one parity byte per eight data bytes alone adds roughly $0.9$\,GB, before verification cost; protecting the identified $3.1$--$5.3\%$ instead raises the tested direct-head budget while touching far less state. Weight shuffling provides no integrity by itself: a static layout may be recovered during memory templating, whereas repeatedly moving gigabytes of weights is bandwidth-heavy. It is therefore complementary to selective protection rather than an equal-cost substitute.

\subsection{Where Damaging Bits Concentrate}\label{ss:underbelly}
The highest-gain bits occupy small action-generating subsets, but their exact layer depends on the architecture.
\paragraph{Localization across architectures.}
On discrete OpenVLA, $K{=}3$ gives deviation 0.105 and $0\%$ closed-loop SR, with top flips in LLM L1. On \pizero{}, $K{=}3$ gives 0.363 deviation, with flips in the action expert; the primary objective favors expert-L17 across 5 noise seeds rather than the vision-language backbone (Fig.~\ref{fig:soft}). Thus high-leverage layers are sparse but architecture dependent.
\begin{figure}[t]\centering
\includegraphics[width=\columnwidth]{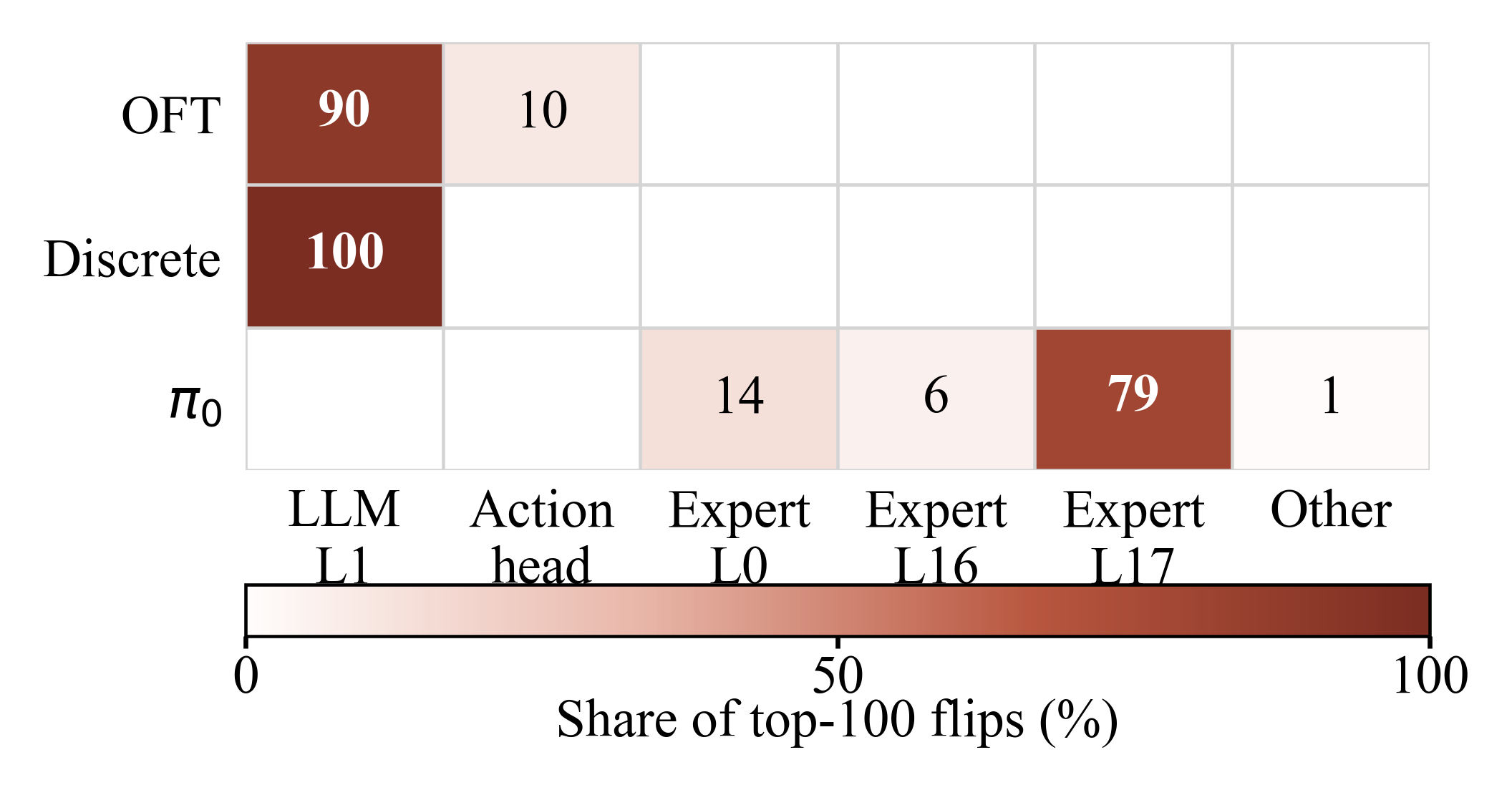}
\caption{Layer localization among the top-100 flips (blank cells are zero). Direct heads concentrate in LLM L1, whereas \pizero{} concentrates in its final expert layer (5-seed mean).}\label{fig:soft}
\end{figure}
This localization result is distinct from the budget spectrum: it identifies \emph{where} an attack enters each model, not why one decoder needs more flips than another.

\subsection{Flow-Head Objective Analysis}\label{ss:flowanalysis}
We next isolate the role of the executed-action objective on the evaluated flow heads.
\paragraph{Complete flow-head budget curves.}
Table~\ref{tab:flowbudgets} makes the primary fixed-direction endpoints explicit. On \pizero{}/Spatial, the attack removes roughly half the clean successes by \(K{=}30\) and reaches \(0\%\) at ${\sim}100$; on Object it leaves only \(3.3\%\) at \(K{=}100\). The distinct $\pi_{0.5}$ reaches \(0\%\) at \(K{=}300\), supporting the empirical ${\sim}100$--$300$ flow-head regime.
\begin{table}[t]\centering\small
\setlength{\tabcolsep}{2.0pt}
\begin{tabular}{lccc}\toprule
policy/suite & clean & \shortstack{directed\\SR@\(K\)} & \shortstack{random\\SR@\(K\)}\\\midrule
\pizero{}/Spatial & 70.0\% & \shortstack{${\sim}\tfrac{1}{2}$ clean@30;\\0.0\%@${\sim}100$} & ${\sim}$clean@300\\
\pizero{}/Object & 86.7\% & 3.3\%@100 & 70.0\%@100\\
$\pi_{0.5}$/Spatial & 93.3\% & \shortstack{66.7\%@30; 30.0\%@100;\\0.0\%@300} & 96.7\%@300\\\bottomrule
\end{tabular}
\caption{Closed-loop flow-head budgets ($n{=}30$).}\label{tab:flowbudgets}
\end{table}
Exact \(n{=}30\) counts for the latter two rows are \pizero{}/Object \(26/1/21\) clean/directed/random at \(K{=}100\), and $\pi_{0.5}$/Spatial \(28/30\) clean, \(20/9/0\) directed at \(K{=}30/100/300\), and \(29/30\) random-\(300\).
For the controlled objective comparison on \pizero{}/Spatial, isotropic energy instead requires ${\sim}1000$ flips to reach \(0\%\) (\(n{=}50\); random 76.0\%), ten times the fixed-direction budget.

\begin{table}[t]\centering\small
\renewcommand{\arraystretch}{0.88}
\begin{tabular}{lccc}\toprule
fixed direction & successes & SR $\downarrow$ & 95\% CI \\\midrule
clean (no flips) & 18/30 & 60.0\% & [40.6, 77.3] \\
$+\mathbf{1}$ & 2/30 & 6.7\% & [0.8, 22.1] \\
$-\mathbf{1}$ & 8/30 & 26.7\% & [12.3, 45.9] \\
random $u_0$ & 2/30 & 6.7\% & [0.8, 22.1] \\
random $u_1$ & 2/30 & 6.7\% & [0.8, 22.1] \\
random $u_2$ & 4/30 & 13.3\% & [3.8, 30.7] \\\bottomrule
\end{tabular}
\caption{Matched fixed-direction \pizero{} trials at $K{=}100$ ($n{=}30$).}\label{tab:directions}
\end{table}

\paragraph{Objective-direction robustness.}
We repeat the $K{=}100$ \pizero{}/LIBERO-Spatial attack with $+\mathbf{1}$, $-\mathbf{1}$, and three independent Gaussian unit directions. Each receives a fresh one-shot ranking; 414 candidate linears, three calibration frames, two noise draws, quantization, and the same 30 task--episode seeds are fixed.

From the matched $18/30$ clean rate, the random directions yield $2/30$, $2/30$, and $4/30$ successes; $-\mathbf{1}$ is weaker ($8/30$; Table~\ref{tab:directions}). After Bonferroni correction across the five directions, $+\mathbf{1}$ and all random directions remain significant (adjusted $p\le2.6{\times}10^{-3}$), whereas $-\mathbf{1}$ does not (raw $p=.031$, adjusted $p=.154$). Thus direction affects strength, but the result is not all-positive-specific. Every ranking stays in the action expert, usually with 81--92 flips in L17; one random direction splits between L0 (44) and L17 (43).

\begin{table}[t]\centering\small
\begin{tabular}{lcc}
\toprule
attack & direct, $K{=}3$ & \pizero{} flow \\
\midrule
clean & 90\% & 70\% \\
random-300 & 90\% & ${\sim}$clean \\
magnitude-3 & 90\% & ${\sim}$clean \\
token-CE & 0\% & n/a \\
self-deviation & n/a & 76.7\% at $K{=}1000$ \\
isotropic $\max\|a\|^2$ & 0\% & 0\% at ${\sim}1000$ \\
manifold escape & 0\% & 0\% at ${\sim}100$ \\
\bottomrule
\end{tabular}
\caption{Head-specific empirical collapse budgets.}
\label{tab:ports}
\end{table}
\paragraph{Attack comparison.} Table~\ref{tab:ports} separates head-specific ports. On the direct head, any gradient-aligned objective collapses at $K{=}3$, whereas random and magnitude baselines remain at clean SR; without iterative decoding, objective design matters little. On \pizero{}, the direct self-deviation port remains at 76.7\% SR at $K{=}1000$ because its clean-initialization gradient is zero. The meaningful nonzero-gradient comparison is therefore isotropic energy versus manifold escape, which require ${\sim}1000$ and ${\sim}100$ flips, respectively.

\paragraph{Step-effect consistency across rerankings.}
Across ten \(K{=}100\) rerankings (51 shared traces each), fixed-direction coherence is \(1.000\) versus \(0.969\) for energy. Differences are positive/tied in \(7/3\) cases (one-sided Wilcoxon \(p{=}.0078\)), with mean \(0.031\) (95\% CI \([-0.019,0.081]\)).

\paragraph{Architecture-level synthesis.}
Together, the results separate \emph{where} faults enter from \emph{how} they propagate: localization guides protection, while decoder architecture and directional consistency shape the empirical closed-loop budget.

\paragraph{Limitations and ethics.}
Reported \(K\) budgets measure logical INT8 susceptibility, not hardware-independent margins; physical fault delivery and ECC remain outside scope. This dual-use study is intended to motivate integrity protection.

\section{Conclusion}
Selected INT8 flips collapse random-fault-tolerant VLAs: 1--5 for direct heads versus ${\sim}100$--$300$ for the evaluated flow heads. Fixed-direction manifold escape cuts \pizero{}'s budget tenfold, works across five directions, and guides protection through $K{=}100$. Physical-robot trials validate task-calibrated failure. This architecture-aware view links attack construction to selective protection. It also motivates decoder-aware fault evaluation before deployment. For direct heads, localized integrity checks offer a practical starting point. Weight integrity is therefore a VLA security boundary.

\newpage
\input{appendix}
\end{document}

%% file: author_block.tex
\author{
Yudong Gao\textsuperscript{\rm 1}\equalcontrib,
Linghan Chen\textsuperscript{\rm 2}\equalcontrib,
Wenhan Wu\textsuperscript{\rm 3},
Mia Zhou\textsuperscript{\rm 4},\\
Jiyao Wang\textsuperscript{\rm 5},
Kaiyan Ji\textsuperscript{\rm 2},
Mingyu Guo\textsuperscript{\rm 2},
Honglong Chen\textsuperscript{\rm 6}\corresponding
}
\affiliations{
\textsuperscript{\rm 1}The Hong Kong University of Science and Technology\\
\textsuperscript{\rm 2}Adelaide University\\
\textsuperscript{\rm 3}Wuhan University\\
\textsuperscript{\rm 4}University of North Carolina at Chapel Hill\\
\textsuperscript{\rm 5}ETH Zurich\\
\textsuperscript{\rm 6}China University of Petroleum (East China)\\
yudonggao0504@163.com
}

%% file: appendix.tex
\appendix
\section{Additional Statistical and Reproducibility Details}

\subsection{Code and environment}
The ancillary code contains reusable INT8 quantization and bit-ranking utilities; implementations for the direct, discrete-token, \pizero{}, and $\pi_{0.5}$ experiments; analysis code for the fixed-direction study; and the real-robot patch-and-restore utility. It intentionally excludes checkpoints, benchmark assets, calibration observations, selected patch lists, robot videos, and experiment outputs. The \pizero{} jobs used one NVIDIA A800 under Ubuntu 22.04.3, Python 3.12.12, PyTorch 2.5.0, Transformers 4.53.3, LeRobot 0.4.4, and LIBERO 0.1.0. The direction sweep fixes seed 20260722, three calibration frames, two ranking-noise draws, and 30 rollout slots per condition. The multi-seed fixed-path study independently reranks seeds 20260730--39, using three calibration frames, two ranking-noise draws, and 17 held-out frames with three shared evaluation-noise draws per seed.

\subsection{Real-robot aggregate results}
Table~\ref{tab:robotstats} aggregates the blue-bowl trials used in the main paper. Confidence intervals are two-sided Clopper--Pearson intervals.
\begin{table}[H]
\centering
\small
\begin{tabular}{lccc}
\toprule
condition & success & SR & exact 95\% CI \\
\midrule
clean & $14/20$ & 70.0\% & $[45.7,88.1]\%$ \\
directed $K{=}100$ & $0/20$ & 0.0\% & $[0.0,16.8]\%$ \\
global-random $K{=}100$ & $16/20$ & 80.0\% & $[56.3,94.3]\%$ \\
\bottomrule
\end{tabular}
\caption{Physical-robot success on the blue-bowl task.}
\label{tab:robotstats}
\end{table}
Two-sided Fisher exact tests give $p=3.34\!\times\!10^{-6}$ for directed versus clean and $p=1.54\!\times\!10^{-7}$ for directed versus global-random; clean versus global-random is not distinguishable ($p=0.716$). These comparisons do not establish task transfer, and the random control is not layer/bit matched.

\subsection{Fixed-path consistency diagnostic}
At $K{=}100$, corrupted velocities are re-evaluated at every clean denoising state, and coherence is $|\sum_k c_k|/\sum_k|c_k|$. A single expanded ranking paired on 17 held-out frames and three shared noise draws ($n{=}51$ traces) gave coherence $1.000$ versus $0.998$ and did not resolve a difference ($p=.159$). Because the independently rebuilt ranking is the relevant experimental unit, we then fixed a ten-seed protocol before inspecting its results. Each seed reranks both objectives and averages the same 51 within-seed traces. Fixed direction has coherence $1.000$ for all ten seeds, whereas energy averages $0.969$: seven paired differences are positive, three are ties, and none is negative (one-sided Wilcoxon $p=.0078$). The mean difference is $0.031$, but its 95\% CI is wide ($[-0.019,0.081]$). Thus the experiment supports a tendency toward more consistent step effects, while the tenfold closed-loop budget gap remains an empirical objective comparison rather than an identified causal mechanism.

Layer identity varies across ranking seeds. Manifold-escape flips usually concentrate in Expert-L17, consistent with the 5-seed localization study, whereas energy often favors Expert-L12 and sometimes Expert-L0. The supported claim is concentration in a small action-expert subset, not invariance of one exact layer.

\subsection{Attack-port comparison}
The direct-head result reduces to gradient-ranked bit search, while the relevant flow-head comparison separates zero-gradient and nonzero-gradient objectives:
\begin{table}[H]
\centering
\small
\begin{tabular}{lcc}
\toprule
attack & direct, $K{=}3$ & \pizero{} flow \\
\midrule
clean & 90\% & 70\% \\
random-300 & 90\% & ${\sim}$clean \\
magnitude-3 & 90\% & ${\sim}$clean \\
token-CE & 0\% & n/a \\
self-deviation & n/a & 76.7\% at $K{=}1000$ \\
isotropic $\max\|a\|^2$ & 0\% & 0\% at ${\sim}1000$ \\
manifold escape & 0\% & 0\% at ${\sim}100$ \\
\bottomrule
\end{tabular}
\caption{Head-specific empirical collapse budgets.}
\label{tab:ports-appendix}
\end{table}

\section{Exploratory Flow-Attenuation Analysis}
These exploratory diagnostics probe why the two evaluated flow-matching policies have high attack budgets. They neither isolate decoder dynamics causally nor predict closed-loop bit budgets.

\subsection{Endpoint sensitivity to initial-noise perturbations}
We perturb initial noise by $\epsilon d$ for random unit $d$ and measure $\rho_{\rm end}=\|\Delta a\|/\|\Delta\,\mathrm{noise}\|$ ($n{=}24$ per $\epsilon$). It is $0.16$--$0.38$ for \pizero{} and $0.04$--$0.22$ for $\pi_{0.5}$ over $\epsilon\in\{0.1,0.3,1.0\}$. These finite endpoint ratios show shrinkage but neither identify the symmetric part of $J_x$ nor verify $J_x+J_x^\top\preceq-2\mu I$; $\kappa_{\rm end}=-\ln\rho_{\rm end}$ is descriptive only. They also do not predict budget ordering: $\pi_{0.5}$ has larger first-order per-flip open-loop damage despite its higher closed-loop budget.

\subsection{Solver-depth diagnostic}
We rebuild each attack for $N\in\{2,5,10,20\}$ denoising steps. Increasing $N$ changes discretization, not the learned vector field, so it is not a controlled intervention on contraction. At fixed budget, isotropic open-loop deviation falls roughly $2\times$ ($0.057$ at $N\leq5$ to $0.025$ at $N{=}20$), while manifold escape collapses success at every $N$. This establishes objective efficacy across solver depths but does not identify cancellation or contractivity; moreover, the deviation result is open-loop and the closed-loop sample ($n{=}12$) is too noisy for a stronger claim.

\section{Conditional First-Order Contraction Bound}
This section records a sufficient-condition analysis that helps organize intuition. The condition is not verified by the endpoint diagnostic above, and the bound is not used to estimate any experimental budget.

\begin{proposition}[Budget lower bound from contraction]\label{prop:budget}
Write denoising in forward time $s\in[0,1]$, with $\dot x_s=v_\theta(x_s,s,c)$ and executed action $a=x_1$. Partition weights into blocks $\theta=\{\theta_\ell\}$ and let $L_\ell:=\sup_s\|\partial_{\theta_\ell}v_\theta\|$. Suppose $J_x=\partial v_\theta/\partial x$ satisfies $J_x+J_x^\top\preceq-2\mu I$ along the sampling trajectory for $\mu>0$. Then the first-order action variation satisfies
\[
\|\delta a\|\le C(\mu)\sum_\ell L_\ell\|\delta\theta_\ell\|,\qquad C(\mu)=\frac{1-e^{-\mu}}{\mu}.
\]
If each of $K$ INT8 flips changes one weight by at most $b$ and $L_{\max}=\max_\ell L_\ell$, reaching first-order open-loop deviation $D$ requires $K\ge D/[L_{\max}bC(\mu)]$.
\end{proposition}
Here $L_\ell$ describes where a perturbation enters the vector field and $C(\mu)$ how a contractive trajectory would attenuate it. Since $C(\mu)\to1$ as $\mu\to0$, the unattenuated first-order case is recovered. This is an open-loop upper bound, not an equality or closed-loop certificate; direct heads lie outside its premise. Task tolerance, bit values, nonlinear interactions, and feedback can dominate, so budgets are measured rather than inferred.

\begin{proof}
For each block, $S_\ell=\partial x_s/\partial\theta_\ell$ obeys $\dot S_\ell=J_xS_\ell+\partial_{\theta_\ell}v$, $S_\ell(0)=0$. The hypothesis gives $\|\Phi(1,s)\|\le e^{-\mu(1-s)}$ and thus $\|S_\ell(1)\delta\theta_\ell\|\le L_\ell C(\mu)\|\delta\theta_\ell\|$. Sum over blocks and use $\sum_\ell\|\delta\theta_\ell\|\le Kb$.
\end{proof}